\documentclass[11pt]{article} % ePrint

\usepackage[letterpaper,hmargin=1in,vmargin=1.25in]{geometry} % ePrint
\usepackage[english]{babel}
\usepackage[utf8]{inputenc}
\usepackage[pdftex,
            pdfauthor={Pablo Rauzy and Sylvain Guilley},
            colorlinks=false,pdfborder={0 0 0},
            pdftitle={Countermeasures Against High-Order Fault-Injection Attacks on CRT-RSA}
            ]{hyperref}
\usepackage{amsmath,amsfonts,amsthm,amssymb,mathtools,stmaryrd,epstopdf}
\usepackage{color,xspace,xcolor,graphicx,paralist,wrapfig}
\graphicspath{{img/}}
\usepackage{array,relsize,alltt,multirow,booktabs,paralist,moreverb}

\usepackage{xmpincl}
\includexmp{CC_Attribution_4.0_International}

\usepackage[ruled,linesnumbered]{algorithm2e}

\SetCommentSty{mycommfont}
\SetAlCapNameFnt{\small}
\SetAlFnt{\small}
\newcommand{\algorithmConfig}{
  \DontPrintSemicolon
	\SetKw{Or}{or}
	\SetKwInOut{Input}{Input}
	\SetKwInOut{Output}{Output}
}
\theoremstyle{definition}
\newtheorem{definition}{Definition}
\newtheorem{remark}{Remark}
\theoremstyle{theorem}
\newtheorem{proposition}{Proposition}
\theoremstyle{lemma}
\newtheorem{lemma}{Lemma}
\theoremstyle{notation}
\newtheorem{notation}{Notation}

\newcommand{\footref}[1]{$^{\ref{#1}}$}

\newcommand{\Z}{\mathbb{Z}}
\newcommand{\ie}{\textit{i.e.}}
\newcommand{\eg}{\textit{e.g.}}
\newcommand{\etal}{\textit{et al.}\xspace}
\newcommand{\etc}{\textit{etc.}\xspace}

\title{
  \texorpdfstring{
    \textbf{Countermeasures Against High-Order\\Fault-Injection Attacks on CRT-RSA}
  }{
    Countermeasures Against High-Order Fault-Injection Attacks on CRT-RSA
  }
}

\author{
  Pablo Rauzy and Sylvain Guilley\\
  Institut Mines-Télécom ; Télécom ParisTech ; CNRS LTCI\\
  \{{\it firstname}.{\it lastname}\}@telecom-paristech.fr
}
\date{} % Removed

\begin{document}
\maketitle

\begin{abstract}
  In this paper we study the existing CRT-RSA countermeasures against fault-injection attacks.
  In an attempt to classify them we get to achieve deep understanding of how they work.
  We show that the many countermeasures that we study (and their variations) actually share a number of common features, but optimize them in different ways.
  We also show that there is no conceptual distinction between test-based and infective countermeasures and how either one can be transformed into the other.
  Furthermore, we show that faults on the code (skipping instructions) can be captured by considering only faults on the data.
  These intermediate results allow us to improve the state of the art in several ways:
  \begin{inparaenum}[(a)]
  \item we fix an existing and that was known to be broken countermeasure (namely the one from Shamir);
  \item we drastically optimize an existing countermeasure (namely the one from Vigilant)
    which we reduce to 3 tests instead of 9 in its original version, and prove that it resists not only one fault but also an arbitrary number of randomizing faults;
  \item we also show how to upgrade countermeasures to resist any given number of faults: given a correct first-order countermeasure, we present a way to design a provable high-order countermeasure (for a well-defined and reasonable fault model).
  \end{inparaenum}
  Finally, we pave the way for a generic approach against fault attacks for any modular arithmetic computations, and thus for the automatic insertion of countermeasures.
\end{abstract}

\section{Introduction}
\label{intro}

Private information protection is a highly demanded feature, especially in the current context of global defiance against most infrastructures, assumed to be controlled by governmental agencies.
Properly used cryptography is known to be a key building block for secure information exchange.
However, in addition to the threat of cyber-attacks, implementation-level hacks must also be considered seriously.
This article deals specifically with the protection of a \emph{decryption} or \emph{signature} crypto-system (called RSA~\cite{DBLP:journals/cacm/RivestSA78}) in the presence of hardware attacks (\eg, we assume the attacker can alter the RSA computation while it is being executed).

It is known since 1997 (with the BellCoRe attack by Boneh \etal~\cite{boneh-fault}) that injecting faults during the computation of CRT-RSA (CRT for ``Chinese Remainder Theorem'') could yield to malformed signatures that expose the prime factors ($p$ and $q$) of the public modulus ($N=p \cdot q$).
Notwithstanding, computing without the fourfold acceleration conveyed by the CRT optimization is definitely not an option in practical applications.
Therefore, many countermeasures have appeared.
Most of the existing countermeasures were designed with an attack-model consisting in a single fault injection.
The remaining few attempts to protect against second-order fault attacks (\ie, attacks with two faults).

Looking at the history of the development of countermeasures against the BellCoRe attack, we see that many countermeasures are actually broken in the first place.
Some of them were fixed by their authors and/or other people, such as the countermeasure proposed by Vigilant~\cite{DBLP:conf/ches/Vigilant08}, which was fixed by Coron \etal~\cite{DBLP:conf/fdtc/CoronGMPV10} and then simplified by Rauzy~\&~Guilley~\cite{PR:PPREW14}; some simply abandoned, such as the one by Shamir~\cite{shamir-patent-rsa-crt}.
Second-order countermeasures are no exception to that rule, as demonstrated with the countermeasure proposed by Ciet~\&~Joye~\cite{cietjoye:fdtc05}, which was fixed later by Dottax \etal~\cite{DBLP:conf/wistp/DottaxGRS09}.
Such mistakes can be explained by two main points:
\begin{compactitem}
\item the almost nonexistent use of formal methods in the field of implementation security, which can itself be explained by the difficulty to properly model the physical properties of an implementation which are necessary to study side-channel leakages and fault-injection effects;
\item the fact that most countermeasures were developed by trial-and-error engineering, accumulating layers of intermediate computations and verifications to patch weaknesses until a fixed point was reached, even if the inner workings of the countermeasure were not fully understood.
\end{compactitem}
Given their development process, it is likely the case that existing second-order countermeasures would not resist third-order attacks, and strengthening them against such attacks using the same methods will not make them resist fourth-order, \etc %include dot

The purpose of this paper is to remedy to these problems.
First-order countermeasures have started to be formally studied by Christofi \etal~\cite{JCEN-Christofi13}, who have been followed by Rauzy~\&~Guilley~\cite{PR:JCEN13,PR:PPREW14}, and Barthe \etal~\cite{DBLP:journals/iacr/BartheDFGTZ14}.
To our best knowledge, no such work has been attempted on high-order countermeasures.
Thus, we should understand the working factors of a countermeasure, and use that knowledge to informedly design a generic high-order countermeasure, either one resisting any number of faults, or one which could be customized to protect against $n$ faults, for any given $n \geqslant 1$.

Notice that we consider RSA used in a mode where the BellCoRe attack is applicable;
this means that we assume that the attacker can choose (but not necessarily knows) the message that is exponentiated,
which is the case in \emph{decryption} mode or in (outdated) \emph{deterministic signature} mode (\eg, PKCS \#1 v1.5).
% http://www.emc.com/collateral/white-papers/h11300-pkcs-1v2-2-rsa-cryptography-standard-wp.pdf (§9.2)
In some other modes, formal proofs of security have been conducted~\cite{DBLP:conf/asiacrypt/CoronM09,DBLP:journals/iacr/BartheDFGTZ14}.

\paragraph*{Contributions}
In this paper we propose a classification of the existing CRT-RSA countermeasures against the BellCoRe fault-injection attacks.
Doing so, we raise questions whose answers lead to a deeper understanding of how the countermeasures work.
We show that the many countermeasures that we study (and their variations) are actually applying a common protection strategy but optimize it in different ways (Sec.~\ref{essence}).
We also show that there is no conceptual distinction between test-based and infective countermeasures and how either one can be transformed into the other (Prop.~\ref{prop:equiv-test-infective}).
Furthermore, we show that faults on the code (skipping instructions) can be captured by considering only faults on the data (Lem.~\ref{lem:equiv-fault-code-data}).
These intermediate results allow us to improve the state of the art in several ways:
\begin{compactitem}
\item we fix an existing and that is known to be broken countermeasure (Alg.~\ref{alg:fixed-shamir});
\item we drastically optimize an existing countermeasure, while at the same time we transform it to be infective instead of test-based (Alg.~\ref{alg:simplified-vigilant});
\item we also show how to upgrade countermeasures to resist any given number of faults: given a correct first-order countermeasure, we present a way to design a provable high-order countermeasure for a well defined and reasonable fault model (Sec.~\ref{essence:order}).
\end{compactitem}
Finally, we pave the way for a generic approach against fault attacks for any modular arithmetic computations, and thus for the automatic insertion of countermeasures.

\paragraph*{Organization of the paper}
We recall the CRT-RSA cryptosystem and the BellCoRe attack in Sec.~\ref{crtrsa-bellcore}.
Then, to better understand the existing countermeasures, we attempt to classify them in Sec.~\ref{classify}, which also presents the state of the art.
We then try to capture what make the essence of a countermeasure in Sec.~\ref{essence}, and use that knowledge to determine how to build a high-order countermeasure.
We last use our findings to build better countermeasures by fixing and simplifying existing ones in Sec.~\ref{building}.
Conclusions and perspectives are drawn in Sec.~\ref{conclusions}.
The appendices contain the detail of some secondary results.

\section{CRT-RSA and the BellCoRe Attack}
\label{crtrsa-bellcore}

This section summarizes known results about fault attacks on CRT-RSA (see also \cite{koc_rsa}, \cite[Chap. 3]{Intro_HOST} and \cite[Chap. 7~\&~8]{fabook}).
Its purpose is to settle the notions and the associated notations that will be used in the later sections, to present our novel contributions.

\subsection{RSA}
\label{sec:RSA}

RSA is both an \emph{encryption} and a \emph{signature} scheme.
It relies on the fact that for any message $0\leq M<N$,
$(M^d)^e \equiv M \mod N$, where $d \equiv e^{-1} \mod \varphi(N)$, by Euler's theorem%
\footnote{We use the usual convention in all mathematical equations, namely that the ``$\text{mod}$'' operator has the lowest binding precedence,
\ie, $a \times b \mod c \times d$ represents the element $a \times b$ in $\Z_{c \times d}$.}.
In this equation, $\varphi$ is Euler's totient function, equal to $\varphi(N) = (p - 1) \cdot (q - 1)$ when $N = p \cdot q$ is a composite number, product of two primes $p$ and $q$.
For example, if Alice generates the signature $S = M^d \mod N$,
then Bob can verify it by computing $S^e \mod N$, which must be equal to $M$ unless Alice is only pretending to know $d$.
Therefore $(N,d)$ is called the private key, and $(N,e)$ the public key.
In this paper, we are not concerned about the key generation step of RSA,
and simply assume that $d$ is an unknown number in $\llbracket 1, \varphi(N)=(p-1) \cdot (q-1)\llbracket$.
Actually, $d$ can also be chosen to be equal to the smallest value $e^{-1} \mod \lambda(N)$,
where $\lambda(N) = \frac{(p-1) \cdot (q-1)}{\gcd(p-1, q-1)}$ is the Carmichael function (see PKCS \#1 v2.1, \S 3.1).
% ftp://ftp.rsasecurity.com/pub/pkcs/pkcs-1/pkcs-1v2-1.pdf

\subsection{CRT-RSA}
\label{sec:CRT_RSA}

The computation of $M^d \mod N$ can be speeded-up by a factor of four using the Chinese Remainder Theorem (CRT).
Indeed, numbers modulo $p$ and $q$ are twice as short as those modulo $N$.
For example, for $2,048$~bits RSA, $p$ and $q$ are $1,024$~bits long.
CRT-RSA consists in computing $S_p = M^d \mod p$ and $S_q = M^d \mod q$,
which can be recombined into $S$ with a limited overhead.
Due to the little Fermat theorem (the special case of the Euler theorem when the modulus is a prime),
$S_p = (M \mod p)^{d \mod (p-1)} \mod p$.
This means that in the computation of $S_p$, the processed data have $1,024$~bits,
and the exponent itself has $1,024$~bits (instead of $2,048$~bits).
Thus the multiplication is four times faster and the exponentiation eight times faster.
However, as there are two such exponentiations (modulo $p$ and $q$), the overall CRT-RSA is roughly speaking four times faster than RSA computed modulo $N$.

This acceleration justifies that CRT-RSA is always used if the factorization of $N$ as $p \cdot q$ is known.
In CRT-RSA, the private key has a richer structure than simply $(N,d)$:
it is actually the $5$-tuple $(p, q, d_p, d_q, i_q)$, where:
\begin{compactitem}
\item $d_p \doteq d \mod (p-1)$,
\item $d_q \doteq d \mod (q-1)$, and
\item $i_q \doteq q^{-1} \mod p$.
\end{compactitem}
\medskip

The CRT-RSA algorithm is presented in Alg.~\ref{alg:unprotected}.
It is straightforward to check that the signature computed at line~\ref{alg:unprotected:S}
belongs to $\llbracket 0, p \cdot q -1 \rrbracket$.
Consequently, no reduction modulo $N$ is necessary before returning $S$.

\begin{algorithm*}
\algorithmConfig
\caption{Unprotected CRT-RSA}
\label{alg:unprotected}
\Input{Message $M$, key $(p, q, d_p, d_q, i_q)$}
\Output{Signature $M^d \mod N$}
\BlankLine
$S_p = M^{d_p} \mod p$ \tcp*[r]{Intermediate signature in $\Z_p$}
$S_q = M^{d_q} \mod q$ \tcp*[r]{Intermediate signature in $\Z_q$}
\BlankLine
$S = S_q + q \cdot (i_q \cdot (S_p - S_q) \mod p)$ \tcp*[r]{Recombination in $\Z_N$ (Garner's method~\cite{DBLP:journals/ac/Garner65})\label{alg:unprotected:S}}
\BlankLine
\Return $S$
\end{algorithm*}

\subsection{The BellCoRe Attack}

In 1997, a dreadful remark has been made by Boneh, DeMillo and Lipton~\cite{boneh-fault}, three staff of Bell Communication Research:
Alg.~\ref{alg:unprotected} could reveal the secret primes $p$ and $q$ if the line 1 or 2 of the computation is faulted, even in a very random way.
The attack can be expressed as the following proposition.

\begin{proposition}[BellCoRe attack]
\label{prop:bellcore}
If the intermediate variable $S_p$ (resp. $S_q$) is returned faulted as $\widehat{S_p}$ (resp. $\widehat{S_q}$)\footnote{In other papers, the faulted variables (such as $X$) are written either as $X^*$ or $\tilde{X}$; in this paper, we use a hat which can stretch to cover the adequate portion of the expression, as it allows to make an unambiguous difference between $\widehat{X}^e$ and $\widehat{X^e}$.},
then the attacker gets an erroneous signature $\widehat{S}$,
and is able to recover $q$ (resp. $p$) as $\gcd(N, S-\widehat{S})$.
\end{proposition}
\begin{proof}
For any integer $x$, $\gcd(N, x)$ can only take $4$ values:
\begin{compactitem}
\item $1$, if $N$ and $x$ are coprime,
\item $p$, if $x$ is a multiple of $p$,
\item $q$, if $x$ is a multiple of $q$,
\item $N$, if $x$ is a multiple of both $p$ and $q$, \ie, of $N$.
\end{compactitem}
\medskip

In Alg.~\ref{alg:unprotected}, if $S_p$ is faulted (\ie, replaced by $\widehat{S_p} \neq S_p$), then
$S-\widehat{S} = q \cdot (( i_q \cdot (S_p-S_q)\mod p) - ( i_q \cdot (\widehat{S_p}-S_q) \mod p))$,
and thus $\gcd(N, S-\widehat{S}) = q$.
If $S_q$ is faulted (\ie, replaced by $\widehat{S_q} \neq S_q$), then
$S-\widehat{S} \equiv (S_q - \widehat{S_q}) - ( q \mod p) \cdot i_q \cdot (S_q - \widehat{S_q}) \equiv 0 \mod p$ because $(q \mod p) \cdot i_q \equiv 1 \mod p$, and thus $S-\widehat{S}$ is a multiple of $p$.
Additionally, $S-\widehat{S}$ is not a multiple of $q$.
So, $\gcd(N, S-\widehat{S})=p$.
\end{proof}

Before continuing to the next section, we will formalize our attack model by defining what is a fault injection and what is the order of an attack.

\begin{definition}[Fault injection]
\label{def:fault}
During the execution of an algorithm, the attacker can:
\begin{compactitem}
\item modify any intermediate value by setting it to either a \emph{random value} (\emph{randomizing fault}) or \emph{zero} (\emph{zeroing fault});
such a fault can be either \emph{permanent} (\eg, in memory) or \emph{transient} (\eg, in a register or a bus);
\item skip any number of consecutive instructions (\emph{skipping fault}).
\end{compactitem}
At the end of the computation the attacker can read the result returned by the algorithm.
\end{definition}

\begin{remark}
This fault injection model implies that faults can be injected very accurately in timing (the resolution is the clock period),
whereas the fault locality in space is poor (the attacker cannot target a specific bit).
This models an attacker who is able to identify the sequence of operations by a simple side-channel analysis,
but who has no knowledge of the chip internals. % to fault a precise location.
Such attack model is realistic for designs where the memories are scrambled and the logic gates randomly routed (in a sea of gates).
\end{remark}

\begin{lemma}
\label{lem:equiv-fault-code-data}
The effect of a skipping fault (\ie, fault on the code) can be captured by considering only randomizing and zeroing faults (\ie, fault on the data).
\end{lemma}
\begin{proof}
Indeed, if the skipped instructions are part of an arithmetic operation:
\begin{compactitem}
\item either the computation has not been done at all and the value in memory where the result is supposed to be stays zero (if initialized) or random (if not),
\item or the computation has partly been done and the value written in memory as its result is thus pseudo-randomized (and considered random at our modeling level).
\end{compactitem}
If the skipped instruction is a branching instruction, then it is equivalent to do a zeroing fault on the result of the branching condition to make it false and thus avoid branching.
\end{proof}

\begin{definition}[Attack order]
\label{def:order}
We call \emph{order} of the attack the number of fault injections in the computation.
An attack is said to be \emph{high-order} if its order is strictly more than $1$.
\end{definition}

\section{Classifying Countermeasures}
\label{classify}

The goal of a countermeasure against fault-injection attacks is to avoid returning a compromised value to the attacker.
To this end, countermeasures attempt to verify the integrity of the computation before returning its result.
If the integrity is compromised, then the returned value should be a random number or an error constant, in order not to leak any information.

An obvious way of achieving that goal is to repeat the computation and compare the results, but this approach is very expensive in terms of computation time.
The same remark applies to the verification of the signature
(notice that $e$ can be recovered for this purpose from the $5$-tuple $(p, q, d_p, d_q, i_q)$, as explained in App.~\ref{sec-e_d}).
In this section we explore the different methods used by the existing countermeasures to verify the computation integrity faster than $(M^d)^e \stackrel{?}{\equiv} M \mod N$.

\subsection{Shamir's or Giraud's Family of Countermeasures}
\label{classify:family}

To the authors knowledge, there are two main families of countermeasures:
those which are descendants of Shamir's countermeasure~\cite{shamir-patent-rsa-crt},
and those which are descendants of Giraud's~\cite{DBLP:journals/tc/Giraud06}.

The countermeasures in Giraud's family avoid replicating the computations using particular exponentiation algorithms.
These algorithms keep track of variables involved in intermediate steps;
those help verifying the consistency of the final results by a consistency check of an invariant that is supposed to be spread till the last steps.
This idea is illustrated in Alg.~\ref{alg:giraud}, which resembles the one of Giraud.
The test at line~\ref{alg:giraud:verif} verifies that the recombined values $S$ and $S'$ (recombination of intermediate steps of the exponentiation) are consistent.
Example of other countermeasures in this family are the ones of Boscher \etal~\cite{DBLP:conf/wistp/BoscherNP07}, Rivain~\cite{cryptoeprint:2009:165} (and its recently improved version~\cite{DBLP:conf/ctrsa/LeRT14}), or Kim \etal~\cite{Kim:2011:ECA:2010601.2010865}.
The former two mainly optimize Giraud's, while the latter introduce an infective verification based on binary masks.
The detailed study of the countermeasures in Giraud's family is left as future work.

\begin{algorithm*}
\algorithmConfig
\caption{CRT-RSA with a Giraud's family countermeasure}
\label{alg:giraud}
\Input{Message $M$, key $(p, q, d_p, d_q, i_q)$}
\Output{Signature $M^d \mod N$, or \textsf{error}}
\BlankLine
$(S_p, S'_p) = \mathsf{ExpAlgorithm}(M, d_p)$ \tcp*[r]{$\mathsf{ExpAlgorithm}(a, b)$ returns $(a^b, a^{b-1})$}
$(S_q, S'_q) = \mathsf{ExpAlgorithm}(M, d_q)$ \;
\BlankLine
$S = S_q + q \cdot (i_q \cdot (S_p - S_q) \mod p)$ \tcp*[r]{Recombination}
$S' = S'_q + q \cdot (i_q \cdot (S'_p - S'_q) \mod p)$ \tcp*[r]{Recombination for verification}
\BlankLine
\lIf{$M \cdot S' \not\equiv S \mod pq$ \label{alg:giraud:verif}}{
  \Return\textsf{error}
}
\BlankLine
\Return $S$
\end{algorithm*}

Indeed, the rest of our paper is mainly concerned with Shamir's family of countermeasures.
The countermeasures in Shamir's family rely on a kind of ``checksum'' of the computation using smaller numbers (the checksum is computed in rings smaller than the ones of the actual computation).
The base-two logarithm of the smaller rings cardinal is typically equal to $32$ or to $64$ (bits):
therefore, assuming that the faults are randomly distributed, the probability of having an undetected fault is $2^{-32}$ or $2^{-64}$,
\ie, very low.
In the sequel, we will make a language abuse by considering that such probability is equal to zero.
We also use the following terminology:
\begin{notation}
Let $a$ a big number and $b$ a small number, such that they are coprime.
We call the ring $\Z_{a b}$ an \emph{overring} of $\Z_a$, and
the ring $\Z_b$ a \emph{subring} of $\Z_{a b}$.
\end{notation}
\begin{remark}
RSA is friendly % Synonymous: Amenable.
to protections by checksums because it computes in rings $\Z_a$ where $a$ is either a large prime number (\eg, $a=p$ or $a=q$) or the product of large prime numbers (\eg, $a=p \cdot q$).
Thus, any small number $b>1$ is coprime with $a$,
and so we have an isomorphism between the \emph{overring} $\Z_{ab}$ and the direct product of $\Z_a$ and $\Z_b$, \ie,
$\Z_{ab} \cong \Z_a \times \Z_b$.
This means that the Chinese Remainder Theorem applies.
Consequently, the nominal computation and the checksum can be conducted in parallel in $\Z_{ab}$.
\end{remark}
The countermeasures attempt to assert that some invariants on the computations and the checksums hold.
There are many different ways to use the checksums and to verify these invariants.
In the rest of this section we review these ways while we attempt to classify countermeasures and understand better what are the necessary invariants to verify.

\subsection{Test-Based or Infective}
\label{classify:test-infec}

A first way to classify countermeasures is to separate
those which consist in step-wise internal checks during the CRT computation
and those which use an infective computation strategy to make the result unusable by the attacker in case of fault injection.

\begin{definition}[Test-based countermeasure]
\label{def:test-based}
A countermeasure is said to be \emph{test-based} if it attempts to detect fault injections by verifying that some arithmetic invariants are respected, and branch to return an error instead of the numerical result of the algorithm in case of invariant violation.
Examples of test-based countermeasures are the ones of Shamir~\cite{shamir-patent-rsa-crt}, Aumüller \etal~\cite{DBLP:conf/ches/AumullerBFHS02}, Vigilant~\cite{DBLP:conf/ches/Vigilant08}, or Joye \etal~\cite{joye01:crt-rsa}.
\end{definition}

\begin{definition}[Infective countermeasure]
\label{def:infective}
A countermeasure is said to be \emph{infective} if rather than testing arithmetic invariants it uses them to compute a neutral element of some arithmetic operation in a way that would not result in this neutral element if the invariant is violated.
It then uses the results of these computations to infect the result of the algorithm before returning it to make it unusable by the attacker (thus, it does not need branching instructions).
Examples of infective countermeasures are the ones by Blömer \etal~\cite{DBLP:conf/ccs/BlomerOS03}, Ciet~\&~Joye~\cite{cietjoye:fdtc05}, or Kim \etal~\cite{Kim:2011:ECA:2010601.2010865}.
\end{definition}

The extreme similarity between the verifications in the test-based countermeasure of Joye \etal~\cite{joye01:crt-rsa} (see Alg.~\ref{alg:joye}, line~\ref{alg:joye:verif}) and the infective countermeasure of Ciet~\&~Joye~\cite{cietjoye:fdtc05} (see Alg.~\ref{alg:cietjoye}, lines~\ref{alg:cietjoye:c1} and~\ref{alg:cietjoye:c2}) is striking, but it is actually not surprising at all, as we will discover in Prop.~\ref{prop:equiv-test-infective}.

\begin{algorithm*}
\algorithmConfig
\caption{CRT-RSA with Joye \etal's countermeasure~\cite{joye01:crt-rsa}}
\label{alg:joye}
\Input{Message $M$, key $(p, q, d_p, d_q, i_q)$}
\Output{Signature $M^d \mod N$, or \textsf{error}}
\BlankLine
Choose two small random integers $r_1$ and $r_2$. \;
Store in memory $p' = p \cdot r_1$, $q' = q \cdot r_2$, $i'_q = q'^{-1} \mod p'$, $N = p \cdot q$. \;
\BlankLine
$S'_p = M^{d_p \mod \varphi(p')} \mod p'$ \tcp*[r]{Intermediate signature in $\Z_{pr_1}$}
$S_{pr} = M^{d_p \mod \varphi(r_1)} \mod r_1$ \tcp*[r]{Checksum in $\Z_{r_1}$}
\BlankLine
$S'_q = M^{d_q \mod \varphi(q')} \mod q'$ \tcp*[r]{Intermediate signature in $\Z_{qr_2}$}
$S_{qr} = M^{d_q \mod \varphi(r_2)} \mod r_2$ \tcp*[r]{Checksum in $\Z_{r_2}$}
\BlankLine
$S_p = S'_p \mod p$ \tcp*[r]{Retrieve intermediate signature in $\Z_p$}
$S_q = S'_q \mod q$ \tcp*[r]{Retrieve intermediate signature in $\Z_q$}
\BlankLine
\lIf{$S'_p \not\equiv S_{pr} \mod r_1$ \Or $S'_q \not\equiv S_{qr} \mod r_2$ \label{alg:joye:verif}}{
  \Return\textsf{error}
}
\BlankLine
\Return $S = S_q + q \cdot (i_q \cdot (S_p - S_q) \mod p)$ \tcp*[r]{Recombination in $\Z_{N}$}
\end{algorithm*}

\begin{algorithm*}
\algorithmConfig
\caption{CRT-RSA with Ciet~\&~Joye's countermeasure~\cite{cietjoye:fdtc05}}
\label{alg:cietjoye}
\Input{Message $M$, key $(p, q, d_p, d_q, i_q)$}
\Output{Signature $M^d \mod N$, or a random value in $\Z_N$}
\BlankLine
Choose small random integers $r_1$, $r_2$, and $r_3$. \;
Choose a random integer $a$. \;
Initialize $\gamma$ with a random number \;
Store in memory $p' = p \cdot r_1$, $q' = q \cdot r_2$, $i'_q = q'^{-1} \mod p'$, $N = p \cdot q$. \;
\BlankLine
$S'_p = a + M^{d_p \mod \varphi(p')} \mod p'$ \tcp*[r]{Intermediate signature in $\Z_{pr_1}$}
$S_{pr} = a + M^{d_p \mod \varphi(r_1)} \mod r_1$ \tcp*[r]{Checksum in $\Z_{r_1}$}
\BlankLine
$S'_q = a + M^{d_q \mod \varphi(q')} \mod q'$ \tcp*[r]{Intermediate signature in $\Z_{qr_2}$}
$S_{qr} = a + M^{d_q \mod \varphi(r_2)} \mod r_2$ \tcp*[r]{Checksum in $\Z_{r_2}$}
\BlankLine
$S' = S'_q + q' \cdot (i'_q \cdot (S'_p - S'_q) \mod p')$ \tcp*[r]{Recombination in $\Z_{Nr_1r_2}$\label{alg:cietjoye:S'}}
\BlankLine
$c_1 = S' - S_{pr} + 1 \mod r_1$ \tcp*[r]{Invariant for the signature modulo $p$\label{alg:cietjoye:c1}}
$c_2 = S' - S_{qr} + 1 \mod r_2$ \tcp*[r]{Invariant for the signature modulo $q$\label{alg:cietjoye:c2}}
$\gamma = (r_3 \cdot c_1 + (2^l - r_3) \cdot c_2) / 2^l$ \label{line:gamma:alg:cietjoye}\tcp*[r]{$\gamma = 1$ if $c_1$ and $c_2$ have value $1$}
\BlankLine
\Return $S = S' - a^\gamma \mod N$ \tcp*[r]{Infection and result retrieval in $\Z_N$}
\end{algorithm*}

\begin{proposition}[Equivalence between test-based and infective countermeasures]
\label{prop:equiv-test-infective}
Each test-based (resp. infective) countermeasure has a direct equivalent infective (resp. test-based) countermeasure.
\end{proposition}
\begin{proof}
The invariants that must be verified by countermeasures are modular equality,
so they are of the form $a \stackrel{?}{\equiv} b \mod m$, where $a$, $b$ and $m$ are arithmetic expressions.

It is straightforward to transform this invariant into a Boolean expression usable in test-based countermeasures:
\texttt{if~a~!=~b~[mod~m] then return}~\textit{error}.

To use it in infective countermeasures, it is as easy to verify the same invariant by computing a value which should be $1$ if the invariant holds:
\texttt{c := a - b + 1 mod m}.
The numbers obtained this way for each invariant can then be multiplied and their product $c^*$,
which is $1$ only if all invariants are respected,
can be used as an exponent on the algorithm's result to infect it if one or more of the tested invariants are violated.
Indeed, when the attacker perform the BellCoRe attack by computing $\gcd(N, S - \widehat{S^{c^*}})$ as defined in Prop.~\ref{prop:bellcore}, then if $c^*$ is not $1$ the attack would not work.
\end{proof}

By Prop.~\ref{prop:equiv-test-infective}, we know that there is an equivalence between test-based and infective countermeasures.
This means that in theory any attack working on one kind of countermeasure will be possible on the equivalent countermeasure of the other kind.
However, we remark that in practice it is harder to do a zeroing fault on an intermediate value (especially if it is the result of a computation with big numbers) in the case of an infective countermeasure, than it is to skip one branching instruction in the case of a test-based countermeasure.
We conclude from this the following rule of thumb:
\emph{it is better to use the infective variant of a countermeasure}.
In addition, it is generally the case that code without branches is safer (think of timing attacks or branch predictor attacks on modern CPUs).

Note that if a fault occurs, $c^*$ is not $1$ anymore and thus the computation time required to compute $S^{c^*}$ might significantly increase.
This is not a security problem, indeed, taking longer to return a randomized value in case of an attack is not different from rapidly returning an error constant without finishing the computation first as it is done in the existing test-based countermeasures.
In the worst case scenario, the additional time would be correlated to the induced fault, but we assume the fault to be controlled by the attacker already.

\subsection{Intended Order}
\label{classify:order}

Countermeasures can be classified depending on their order, \ie, the maximum order of the attacks (as per Def.~\ref{def:order}) that they can protect against.

In the literature concerning CRT-RSA countermeasures against fault-injection attacks, most countermeasures claim to be first-order, and a few claim second-order resistance.
For instance, the countermeasures by Aumüller \etal~\cite{DBLP:conf/ches/AumullerBFHS02} and the one by Vigilant~\cite{DBLP:conf/ches/Vigilant08} are described as first-order by their authors,
while Ciet~\&~Joye~\cite{cietjoye:fdtc05} describe a second-order fault model and propose a countermeasure which is supposed to resist to this fault model, and thus be second-order.

However, using the finja\footnote{\url{http://pablo.rauzy.name/sensi/finja.html} (we used the commit \texttt{782384a} version of the code).\label{fn:finja}} tool which has been open-sourced by Rauzy~\&~Guilley~\cite{PR:JCEN13}, we found out that the countermeasure of Ciet~\&~Joye is in fact vulnerable to second-order attacks (in our fault model of Def.~\ref{def:fault}).
This is not very surprising.
Indeed, Prop.~\ref{prop:equiv-test-infective} proves that injecting a fault, and then skipping the invariant verification which was supposed to catch the first fault injection, is a second-order attack strategy which also works for infective countermeasures, except the branching-instruction skip has to be replaced by a zeroing fault.
As expected, the attacks we found using finja did exactly that.
For instance a zeroing fault on $S'_p$ (resp. $S'_q$) makes the computation vulnerable to the BellCoRe attack, and a following zeroing fault on $S_{pr}$ (resp. $S_{qr}$) makes the verification pass anyway.
To our knowledge our attack is new.
It is indeed different from the one Dottax \etal~\cite{DBLP:conf/wistp/DottaxGRS09} found and fixed in their paper, which was an attack on the use of $\gamma$ (see line~\ref{line:gamma:alg:cietjoye} of Alg.~\ref{alg:cietjoye}).
It is true that their attack model only allows skipping faults (as per Def.~\ref{def:fault}) for the second injection, but we have concerns about this:
\begin{compactitem}
\item What justifies this limitation on the second fault? Surely if the attackers are able to inject two faults and can inject a zeroing fault once they can do it twice.
\item Even considering their attack model, a zeroing fault on an intermediate variable $x$ can in many cases be obtained by skipping the instructions where the writing to $x$ happens.
\item The fixed version of the countermeasure by Dottax \etal~\cite[Alg.~8, p.~13]{DBLP:conf/wistp/DottaxGRS09} makes it even closer to the one of Joye \etal by removing the use of $a$ and $\gamma$.
  It also removes the result infection part and instead returns $S$ along with values that should be equal if no faults were injected, leaving ``out'' of the algorithm the necessary comparison and branching instructions which are presented in a separate procedure~\cite[Proc.~1, p.~11]{DBLP:conf/wistp/DottaxGRS09}.
  The resulting countermeasure is second-order resistant (in their attack model) only because the separate procedure does the necessary tests twice (it would indeed break at third-order unless an additional repetition of the test is added, \etc).
\end{compactitem}

An additional remark would be that the algorithms of intended second-order countermeasures does not look very different from others.
Moreover, Rauzy~\&~Guilley~\cite{PR:JCEN13,PR:PPREW14} exposed evidence that the intendedly first-order countermeasures of Aumüller \etal and Vigilant actually offer the same level of resistance against second-order attacks, \ie, they resist when the second injected fault is a randomizing fault (or a skipping fault which amounts to a randomizing fault).

\subsection{Usage of the Small Rings}
\label{classify:smallrings}

In most countermeasures, the computation of the two intermediate signatures modulo $p$ and modulo $q$ of the CRT actually takes place in overrings.
The computation of $S_p$ (resp. $S_q$) is done in $\Z_{pr_1}$ (resp. $\Z_{qr_2}$) for some small random number $r_1$ (resp. $r_2$) rather than in $\Z_p$ (resp. $\Z_q$).
This allows the retrieval of the results by reducing modulo $p$ (resp. $q$) and verifying the signature modulo $r_1$ (resp. $r_2$),
or, if it is done after the CRT recombination, the results can be retrieved by reducing modulo $N = p \cdot q$.
The reduction in the \emph{small subrings} $\Z_{r_1}$ and $\Z_{r_2}$ is used as the checksums for verifying the integrity of the computation.
It works because small random numbers are necessarily coprime with a big prime number.

An interesting part of countermeasures is how they use the small subrings to verify the integrity of the computations.
Almost all the various countermeasures we studied had different ways of using them.
However, they can be divided in two groups.
On one side there are countermeasures which use the small subrings to verify the integrity of the intermediate CRT signatures and of the recombination directly but using smaller numbers, like Blömer \etal's countermeasure~\cite{DBLP:conf/ccs/BlomerOS03}, or Ciet~\&~Joye's one~\cite{cietjoye:fdtc05}.
On the other side, there are countermeasures which use some additional arithmetic properties to verify the necessary invariants indirectly in the small subrings.
Contrary to the countermeasures in the first group, the ones in the second group use the same value $r$ for $r_1$ and $r_2$.
The symmetry obtained with $r_1 = r_2$ is what makes the additional arithmetic properties hold, as we will see.

\subsubsection{Verification of the Intermediate CRT Signatures}
The countermeasure of Blömer \etal~\cite{DBLP:conf/ccs/BlomerOS03} uses the small subrings to verify the intermediate CRT signatures.
It is exposed in Alg.~\ref{alg:bos}.
This countermeasure needs access to $d$ directly rather than $d_p$ and $d_q$ as the standard interface for CRT-RSA suggests, in order to compute $d'_p = d \mod \varphi(p \cdot r_1)$ and $d'_q = d \mod \varphi(q \cdot r_2)$, as well as their inverse $e'_p = {d'_p}^{-1} \mod \varphi(p \cdot r_1)$ and $e'_q = {d'_q}^{-1} \mod \varphi(q \cdot r_2)$ to verify the intermediate CRT signatures.

We can see in Alg.~\ref{alg:bos} that these verifications (lines~\ref{alg:bos:c1} and~\ref{alg:bos:c2}) happen after the recombination (line~\ref{alg:bos:S'}) and retrieve the checksums in $\Z_{r_1}$ (for the $p$ part of the CRT) and $\Z_{r_2}$ (for the $q$ part) from the recombined value $S'$.
It allows these tests to verify the integrity of the recombination at the same time as they verify the integrity of the intermediate CRT signatures.

\begin{algorithm*}
\algorithmConfig
\caption{CRT-RSA with Blömer \etal's countermeasure~\cite{DBLP:conf/ccs/BlomerOS03}}
\label{alg:bos}
\Input{Message $M$, key $(p, q, d, i_q)$}
\Output{Signature $M^d \mod N$, or a random value in $\Z_N$}
\BlankLine
Choose two small random integers $r_1$ and $r_2$. \;
Store in memory $p' = p \cdot r_1$, $q' = q \cdot r_2$, $i'_q = q'^{-1} \mod p'$, $N = p \cdot q$, $N' = N \cdot r_1 \cdot r_2$, $d'_p$, $d'_q$, $e'_p$, $e'_q$. \;
\BlankLine
$S'_p = M^{d'_p} \mod p'$ \tcp*[r]{Intermediate signature in $\Z_{pr_1}$}
$S'_q = M^{d'_q} \mod q'$ \tcp*[r]{Intermediate signature in $\Z_{qr_2}$}
\BlankLine
$S' = S'_q + q' \cdot (i'_q \cdot (S'_p - S'_q) \mod p')$ \tcp*[r]{Recombination in $\Z_{Nr_1r_2}$\label{alg:bos:S'}}
\BlankLine
$c_1 = M - S'^{e'_p} + 1 \mod r_1$ \tcp*[r]{Invariant for the signature modulo $p$\label{alg:bos:c1}}
$c_2 = M - S'^{e'_q} + 1 \mod r_2$ \tcp*[r]{Invariant for the signature modulo $q$\label{alg:bos:c2}}
\BlankLine
\Return $S = S'^{c_1c_2} \mod N$ \tcp*[r]{Infection and result retrieval in $\Z_N$}
\end{algorithm*}

\subsubsection{Checksums of the Intermediate CRT Signatures}
The countermeasure of Ciet~\&~Joye~\cite{cietjoye:fdtc05} uses the small subrings to compute checksums of the intermediate CRT signatures.
It is exposed in Alg.~\ref{alg:cietjoye}.
Just as the previous one, the verifications (lines~\ref{alg:cietjoye:c1} and~\ref{alg:cietjoye:c2}) take place after the recombination (line~\ref{alg:cietjoye:S'}) and retrieve the checksums in $\Z_{r_1}$ (for the $p$ part of the CRT) and $\Z_{r_2}$ (for the $q$ part) from the recombined value $S'$, which enables the integrity verification of the recombination at the same time as the integrity verifications of the intermediate CRT signatures.

We note that this is missing from the protection of Joye \etal~\cite{joye01:crt-rsa}, presented in Alg.~\ref{alg:joye}, which does not verify the integrity of the recombination at all and is thus as broken as Shamir's countermeasure~\cite{shamir-patent-rsa-crt}.
The countermeasure of Ciet~\&~Joye is a clever fix against the possible fault attacks on the recombination of Joye \etal's countermeasure, which also uses the transformation that we described in Prop.~\ref{prop:equiv-test-infective} from a test-based to an infective countermeasure.

\subsubsection{Overrings for CRT Recombination}
In Ciet~\&~Joye's countermeasure the CRT recombination happens in an overring $\Z_{Nr_1r_2}$ of $\Z_N$ while Joye \etal's countermeasure extracts in $\Z_p$ and $\Z_q$ the results $S_p$ and $S_q$ of the intermediate CRT signatures to do the recombination in $\Z_N$ directly.

There are only two other countermeasures which do the recombination in $\Z_N$ that we know of: the one of Shamir~\cite{shamir-patent-rsa-crt} and the one of Aumüller \etal~\cite{DBLP:conf/ches/AumullerBFHS02}.
The first one is known to be broken, in particular because it does not check whether the recombination has been faulted at all.
The second one seems to need to verify $5$ invariants to resist the BellCoRe attack%
\footnote{The original Aumüller \etal's countermeasure uses $7$ verifications because it also needs to check the integrity of intermediate values introduced against simple power analysis, see~\cite[Remark 1]{PR:JCEN13}.},
which is more than the only $2$ required by the countermeasure of Ciet~\&~Joye~\cite{cietjoye:fdtc05} or by the one of Blömer \etal~\cite{DBLP:conf/ccs/BlomerOS03}, while offering a similar level of protection (see \cite{PR:JCEN13}).
This fact led us to think that the additional tests are necessary because the recombination takes place ``in the clear''.
But we did not jump right away to that conclusion.
Indeed, Vigilant's countermeasure~\cite{DBLP:conf/ches/Vigilant08} does the CRT recombination in the $\Z_{Nr^2}$ overring of $\Z_N$ and seems to require $7$ verifications%
\footnote{Vigilant's original countermeasure and its corrected version by Coron \etal~\cite{DBLP:conf/fdtc/CoronGMPV10} actually use $9$ verifications but were simplified by Rauzy~\&~Guilley~\cite{PR:PPREW14} who removed $2$ verifications.}
to also offer that same level of security (see \cite{PR:PPREW14}).
However, we remark that Shamir's, Aumüller \etal's, and Vigilant's countermeasures use the same value for $r_1$ and $r_2$.

\subsubsection{Identity of  $r_1$ and $r_2$}
Some countermeasures, such as the ones of Shamir~\cite{shamir-patent-rsa-crt}, Aumüller \etal~\cite{DBLP:conf/ches/AumullerBFHS02}, and Vigilant~\cite{DBLP:conf/ches/Vigilant08} use a single random number $r$ to construct the overrings used for the two intermediate CRT signatures computation.
The resulting symmetry allows these countermeasures to take advantage of some additional arithmetic properties.

\paragraph{Shamir's countermeasure}
In his countermeasure, which is presented in Alg.~\ref{alg:shamir}, Shamir uses a clever invariant property to verify the integrity of both intermediate CRT signatures in a single verification step (line~\ref{alg:shamir:verif}).
This is made possible by the fact that he uses $d$ directly instead of $d_p$ and $d_q$, and thus the checksums in $\Z_r$ of both the intermediate CRT signatures are supposed to be equal if no fault occurred.
% TODO: SG, the appendix on recovering d from d_p and d_q %% DONE
Unfortunately, the integrity of the recombination is not verified at all.
We will see in Sec.~\ref{building:correction} how to fix this omission.
Besides, we notice that $d$ can be reconstructed from a usual CRT-RSA key $(p, q, d_p, d_q, i_q)$;
we refer the reader to Appendix~\ref{sec-e_d}.

\begin{algorithm*}
\algorithmConfig
\caption{CRT-RSA with Shamir's countermeasure~\cite{shamir-patent-rsa-crt}}
\label{alg:shamir}
\Input{Message $M$, key $(p, q, d, i_q)$}
\Output{Signature $M^d \mod N$, or \textsf{error}}
\BlankLine
Choose a small random integer $r$. \;
\BlankLine
$p' = p \cdot r$ \;
$S'_p = M^{d \mod \varphi(p')} \mod p'$ \tcp*[r]{Intermediate signature in $\Z_{pr}$\label{alg:shamir:S'p}}
\BlankLine
$q' = q \cdot r$ \;
$S'_q = M^{d \mod \varphi(q')} \mod q'$ \tcp*[r]{Intermediate signature in $\Z_{qr}$\label{alg:shamir:S'q}}
\BlankLine
$S_p = S'_p \mod p$ \tcp*[r]{Retrieve intermediate signature in $\Z_p$}
$S_q = S'_q \mod q$ \tcp*[r]{Retrieve intermediate signature in $\Z_q$}
\BlankLine
$S = S_q + q \cdot (i_q \cdot (S_p - S_q) \mod p)$ \tcp*[r]{Recombination in $\Z_N$}
\BlankLine
\lIf{$S'_p \not\equiv S'_q \mod r$ \label{alg:shamir:verif}}{
  \Return\textsf{error}
}
\BlankLine
\Return $S$
\end{algorithm*}

\paragraph{Aumüller \etal's countermeasure}
Contrary to Shamir, Aumüller \etal do verify the integrity of the recombination in their countermeasure, which is presented in Alg.~\ref{alg:aumuller}.
To do this, they straightforwardly check (line~\ref{alg:aumuller:verifS}) that when reducing the result $S$ of the recombination modulo $p$ (resp. $q$), the obtained value corresponds to the intermediate signature in $\Z_p$ (resp. $\Z_q$).
However, they do not use $d$ directly but rather conform to the standard CRT-RSA interface by using $d_p$ and $d_q$.
Thus, they need another verification to check the integrity of the intermediate CRT signatures.
Their clever strategy is to verify that the checksums of $S_p$ and $S_q$ in $\Z_r$ are conform to each other (lines~\ref{alg:aumuller:c_begin} to~\ref{alg:aumuller:c_end}).
For that they check whether ${S_{p}}^{d_{q}}$ is equal to ${S_{q}}^{d_{p}}$ in $\Z_r$, that is, whether the invariant $(M^{d_p})^{d_q} \equiv (M^{d_q})^{d_p} \mod r$ holds.

\begin{algorithm*}
\algorithmConfig
\caption{CRT-RSA with Aumüller \etal's countermeasure\protect\footnotemark~\cite{DBLP:conf/ches/AumullerBFHS02}}
\label{alg:aumuller}
\Input{Message $M$, key $(p, q, d_p, d_q, i_q)$}
\Output{Signature $M^d \mod N$, or \textsf{error}}
\BlankLine
Choose a small random integer $r$. \;
\BlankLine
$p' = p \cdot r$ \;
$q' = q \cdot r$ \;
\lIf{$p' \not\equiv 0 \mod p$ \Or $q' \not\equiv 0 \mod q$ \label{alg:aumuller:verifp'q'}}{
  \Return\textsf{error}
}
\BlankLine
$S'_p = M^{d_p \mod \varphi(p')} \mod p'$ \tcp*[r]{Intermediate signature in $\Z_{pr}$}
$S'_q = M^{d_q \mod \varphi(q')} \mod q'$ \tcp*[r]{Intermediate signature in $\Z_{qr}$}
\BlankLine
$S_p = S'_p \mod p$ \tcp*[r]{Retrieve intermediate signature in $\Z_p$}
$S_q = S'_q \mod q$ \tcp*[r]{Retrieve intermediate signature in $\Z_q$}
\BlankLine
$S = S_q + q \cdot (i_q \cdot (S_p - S_q) \mod p)$ \tcp*[r]{Recombination in $\Z_N$}
\lIf{$S \not\equiv S'_p \mod p$ \Or $S \not\equiv S'_q \mod q$ \label{alg:aumuller:verifS}}{
  \Return\textsf{error}
}
\BlankLine
$S_{pr} = S'_p \mod r$ \tcp*[r]{Checksum of $S_p$ in $\Z_{r}$\label{alg:aumuller:c_begin}}
$S_{qr} = S'_q \mod r$ \tcp*[r]{Checksum of $S_q$ in $\Z_{r}$}
\lIf{${S_{pr}}^{d_q \mod \varphi(r)} \not\equiv {S_{qr}}^{d_p \mod \varphi(r)} \mod r$}{
  \Return\textsf{error} \label{alg:aumuller:c_end}
}
\BlankLine
\Return $S$
\end{algorithm*}

The two additional tests on line~\ref{alg:aumuller:verifp'q'} verify the integrity of $p'$ and $q'$.
Indeed, if $p$ or $q$ happen to be randomized when computing $p'$ or $q'$ the invariant verifications in $\Z_r$ would pass but the retrieval of the intermediate signatures in $\Z_p$ or $\Z_q$ would return random values, which would make the BellCoRe attack work.
These important verifications are missing from all the previous countermeasures in Shamir's family.

\paragraph{Vigilant's countermeasure}
Vigilant takes another approach.
Rather than doing the integrity verifications on ``direct checksums'' that are the representative values of the CRT-RSA computation in the small subrings,
Vigilant uses different values that he constructs for that purpose.
The clever idea of his countermeasure is to use sub-CRTs on the values that the CRT-RSA algorithm manipulates in order to have in one part the value we are interested in and in the other the value constructed for the verification (lines~\ref{alg:vigilant:crtM'p} and~\ref{alg:vigilant:crtM'q}).

\begin{algorithm*}
\algorithmConfig
\caption{CRT-RSA with Vigilant's countermeasure\footref{fn:no-spa}~\cite{DBLP:conf/ches/Vigilant08}\protect\\ {\smaller with Coron \etal's fixes~\cite{DBLP:conf/fdtc/CoronGMPV10} and Rauzy~\&~Guilley's simplifications~\cite{PR:PPREW14}}}
\label{alg:vigilant}
\Input{Message $M$, key $(p, q, d_p, d_q, i_q)$}
\Output{Signature $M^d \mod N$, or \textsf{error}}
\BlankLine
Choose small random integers $r$, $R_1$, and $R_2$. \;
$N = p \cdot q$ \label{alg:vigilant:N_eq_pq} \;
\BlankLine
$p' = p \cdot r^2$ \;
$i_{pr} = p^{-1} \mod r^2$ \;
$M_p = M \mod p'$ \;
$B_p = p \cdot i_{pr}$ \;
$A_p = 1 - B_p \mod p'$ \;
$M'_p = A_p \cdot M_p + B_p \cdot (1 + r) \mod p'$ \tcp*[r]{CRT insertion of verification value in $M'_p$\label{alg:vigilant:crtM'p}}
\BlankLine
$S'_p = {M'_p}^{d_p \mod \varphi(p')} \mod p'$ \tcp*[r]{Intermediate signature in $\Z_{pr^2}$}
\BlankLine
\lIf{$M'_p \not\equiv M \mod p$ \label{alg:vigilant:verifM'p}}{
  \Return\textsf{error}
}
\lIf{$B_p \cdot S'_p \not\equiv B_p \cdot (1 + d_p \cdot r) \mod p'$ \label{alg:vigilant:verifSpr}}{
  \Return\textsf{error}
}
\BlankLine
$q' = q \cdot r^2$ \;
$i_{qr} = q^{-1} \mod r^2$ \;
$M_q = M \mod q'$ \;
$B_q = q \cdot i_{qr}$ \;
$A_q = 1 - B_q \mod q'$ \;
$M'_q = A_q \cdot M_q + B_q \cdot (1 + r) \mod q'$ \tcp*[r]{CRT insertion of verification value in $M'_q$\label{alg:vigilant:crtM'q}}
\BlankLine
$S'_q = {M'_q}^{d_q \mod \varphi(q')} \mod q'$ \tcp*[r]{Intermediate signature in $\Z_{qr^2}$}
\BlankLine
\lIf{$M'_q \not\equiv M \mod q$ \label{alg:vigilant:verifM'q}}{
  \Return\textsf{error}
}
\lIf{$B_q \cdot S'_q \not\equiv B_q \cdot (1 + d_q \cdot r) \mod q'$ \label{alg:vigilant:verifSqr}}{
  \Return\textsf{error}
}
\BlankLine
$S_{pr} = S'_p - B_p \cdot (1 + d_p \cdot r - R_1)$ \tcp*[r]{Verification value of $S'_p$ swapped with $R_1$\label{alg:vigilant:R1}}
$S_{qr} = S'_q - B_q \cdot (1 + d_q \cdot r - R_2)$ \tcp*[r]{Verification value of $S'_q$ swapped with $R_2$\label{alg:vigilant:R2}}
\BlankLine
$S_r = S_{qr} + q \cdot (i_q \cdot (S_{pr} - S_{qr}) \mod p')$ \tcp*[r]{Recombination in $\Z_{Nr^2}$\label{alg:vigilant:LastRecombination}}
\BlankLine
\tcp{Simultaneous verification of lines~\ref{alg:vigilant:N_eq_pq} and \ref{alg:vigilant:LastRecombination}}
\lIf{$pq \cdot (S_r - R_2 - q \cdot i_q \cdot (R_1 - R_2)) \not\equiv 0 \mod N r^2$ \label{alg:vigilant:verifS}}{
  \Return\textsf{error}
}
\BlankLine
\Return $S = S_r \mod N$ \tcp*[r]{Retrieve result in $\Z_N$}
\end{algorithm*}

To do this, he transforms $M$ into another value $M'$ such that:
\[ M' \equiv
\begin{cases}
  M \mod N,\\
  1+r \mod r^2,
\end{cases}
\]
which implies that:
\[ S' = M'^{d} \mod Nr^2 \equiv
\begin{cases}
  M^d \mod N,\\
  1+dr \mod r^2.
\end{cases}
\]
The latter results are based on the binomial theorem, which states that
$ (1+r)^d = \sum_{k=0}^d {d \choose k}r^k = 1+dr + {d \choose 2} r^2 + \text{…}$,
which simplifies to $1+dr$ in the $\Z_{r^2}$ ring.

This property is used to verify the integrity of the intermediate CRT signatures on lines~\ref{alg:vigilant:verifSpr} and~\ref{alg:vigilant:verifSqr}.
It is also used on line~\ref{alg:vigilant:verifS} which tests the recombination using the same technique but with random values inserted on lines~\ref{alg:vigilant:R1} and~\ref{alg:vigilant:R2} in place of the constructed ones.
This test also verifies the integrity of $N$.

Two additional tests are required by Vigilant's arithmetic trick.
The verifications at lines~\ref{alg:vigilant:verifM'p} and~\ref{alg:vigilant:verifM'q} ensure that the original message $M$ has indeed been CRT-embedded in $M'_p$ and $M'_q$.

% Aumüller algorithm footnote, here for presentation purpose
\footnotetext{For the sake of simplicity we removed some code that served against SPA (simple power analysis) and only kept the necessary code against fault-injection attacks. \label{fn:no-spa}}

\section{The Essence of a Countermeasure}
\label{essence}

Our attempt to classify the existing countermeasures provided us with a deep understanding of how they work.
To ensure the integrity of the CRT-RSA computation, the algorithm must verify $3$ things: the integrity of the computation modulo $p$, the integrity of the computation modulo $q$, and the integrity of the CRT recombination (which can be subject to transient fault attacks).
This fact has been known since the first attacks on Shamir's countermeasure.
Our study of the existing countermeasures revealed that, as expected, those which perform these three integrity verifications are the ones which actually work.
This applies to Shamir's family of countermeasures, but also for Giraud's family.
Indeed, countermeasures in the latter also verify the two exponentiations and the recombination by testing the consistency of the exponentiations indirectly on the recombined value.

\subsection{A Straightforward Countermeasure}
\label{essence:straightforward}

The result of these observations is a very straightforward countermeasure, presented in Alg.~\ref{alg:straightforward}.
This countermeasure works by testing the integrity of the signatures modulo $p$ and $q$ by replicating the computations (lines~\ref{alg:straightforward:Sp} and~\ref{alg:straightforward:Sq}) and comparing the results, and the integrity of the recombination by verifying that the two parts of the CRT can be retrieved from the final result (line~\ref{alg:straightforward:S}).
This countermeasure is of course very expensive since the two big exponentiations are done twice, and is thus not usable in practice.
Note that it is nonetheless still better in terms of speed than computing RSA without the CRT optimization.

\begin{algorithm*}
\algorithmConfig
\caption{CRT-RSA with straightforward countermeasure}
\label{alg:straightforward}
\Input{Message $M$, key $(p, q, d_p, d_q, i_q)$}
\Output{Signature $M^d \mod N$, or \textsf{error}}
\BlankLine
$S_p = M^{d_p \mod \varphi(p)} \mod p$ \tcp*[r]{Intermediate signature in $\Z_p$\label{alg:straightforward:Sp}}
\lIf{$S_p \not\equiv M^{d_p} \mod p$ \label{alg:straightforward:verifSp}}{
  \Return \textsf{error}
}
\BlankLine
$S_q = M^{d_q \mod \varphi(q)} \mod q$ \tcp*[r]{Intermediate signature in $\Z_q$\label{alg:straightforward:Sq}}
\lIf{$S_q \not\equiv M^{d_q} \mod q$ \label{alg:straightforward:verifSq}}{
  \Return \textsf{error}
}
\BlankLine
$S = S_q + q \cdot (i_q \cdot (S_p - S_q) \mod p)$ \tcp*[r]{Recombination in $\Z_N$\label{alg:straightforward:S}}
\lIf{$S \not\equiv S_p \mod p$ \Or $S \not\equiv S_q \mod q$ \label{alg:straightforward:verifS}}{
  \Return \textsf{error}
}
\BlankLine
\Return $S$
\end{algorithm*}

\begin{proposition}[Correctness]
\label{prop:correctness}
The straightforward countermeasure (and thus all the ones which do equivalent verifications) is secure against first-order fault attacks as per Def.~\ref{def:fault} and~\ref{def:order}.
\end{proposition}

\begin{proof}
The proof is in two steps.
First, prove that if the intermediate signatures are not correct, then the tests at lines~\ref{alg:straightforward:verifSp} and~\ref{alg:straightforward:verifSq} returns {\smaller \textsf{error}}.
Second, prove that if both tests passed then either the recombination is correct or the test at line~\ref{alg:straightforward:verifS} returns {\smaller \textsf{error}}.

If a fault occurs during the computation of $S_p$ (line~\ref{alg:straightforward:Sp}), then it either has the effect of zeroing its value or randomizing it, as shown by Lem.~\ref{lem:equiv-fault-code-data}.
Thus, the test of line~\ref{alg:straightforward:verifSp} detects it since the two compared values won't be equal.
If the fault happens on line~\ref{alg:straightforward:verifSp}, then either we are in a symmetrical case: the repeated computation is faulted, or the test is skipped: in that case there are no faults affecting the data so the test is unnecessary anyway.
It works similarly for the intermediate signature in $\Z_q$.

If the first two tests pass, then the tests at line~\ref{alg:straightforward:verifS} verify that both parts of the CRT computation are indeed correctly recombined in $S$.
If a fault occurs during the recombination on line~\ref{alg:straightforward:S} it will thus be detected.
If the fault happens at line~\ref{alg:straightforward:verifS}, then either it is a fault on the data and one of the two tests returns {\smaller \textsf{error}}, or it is a skipping fault which bypasses one or both tests but in that case there are no faults affecting the data so the tests are unnecessary anyway.
\end{proof}

\subsection{High-Order Countermeasures}
\label{essence:order}

Using the finja\footref{fn:finja} tool we were able to verify that removing one of the three integrity checks indeed breaks the countermeasure against first-order attacks.
Nonetheless, each countermeasure which has these three integrity checks, plus those that may be necessary to protect optimizations on them, offers the same level of protection.

\begin{proposition}[High-order countermeasures]
\label{prop:high-order}
Against randomizing faults, all correct countermeasures (as per Prop.~\ref{prop:correctness}) are high-order.
However, there are no generic high-order countermeasures if the three types of faults in our attack model are taken into account, but it is possible to build $n$th-order countermeasures for any $n$.
\end{proposition}
\begin{proof}
Indeed, if a countermeasure is able to detect a single randomizing fault, then adding more faults will not break the countermeasure, since a random fault cannot induce a verification skip.
Thus, \emph{all working countermeasures are high-order against randomizing faults}.

However, if after one or more faults which permit an attack, there is a skipping fault or a zeroing fault which leads to skip the verification which would detect the previous fault injections, then the attack will work.
As Lem.~\ref{lem:equiv-fault-code-data} and Prop.~\ref{prop:equiv-test-infective} explain, this is true for all countermeasures, not only those which are test-based but also the infective ones.
It seems that the only way to protect against that is to replicate of the integrity checks.
If each invariant is verified $n$ times, then the countermeasure will resist at least $n$ faults in the worst case scenario: a single fault is used to break the computation and the $n$ others to avoid the verifications which detect the effect of the first fault.
Thus, \emph{there are no generic high-order countermeasures} if the three types of faults in our attack model are taken into account, \textbf{\emph{but it is possible to build a $n\mathrm{th}$-order countermeasure for any $n$} by replicating the invariant verifications $n$ times}.
\end{proof}

Existing first-order countermeasures such as the ones of Aumüller \etal~(Alg.~\ref{alg:aumuller},~\ref{alg:infective-aumuller}), Vigilant~(Alg.~\ref{alg:vigilant},~\ref{alg:simplified-vigilant}), or Ciet~\&~Joye~(Alg.~\ref{alg:cietjoye}) can thus be transformed into $n\mathrm{th}$-order countermeasures, in the attack model described in Def.~\ref{def:fault} and~\ref{def:order}.
As explained, the transformation consists in replicating the verifications $n$ times, whether they are test-based or infective.

This result means that it is very important that the verifications be cost effective.
Fortunately, as we saw in Sec.~\ref{classify} and particularly in Sec.~\ref{classify:smallrings} on the usage of the small rings, the existing countermeasures offer exactly that: optimized versions of Alg.~\ref{alg:straightforward} that use a variety of invariant properties to avoid replicating the two big exponentiations of the CRT computation.

\section{Building Better or Different Countermeasures}
\label{building}

In the two previous sections we learned a lot about current countermeasures and how they work.
We saw that to reduce their cost, most countermeasures use invariant properties to optimize the verification speed by using checksums on smaller numbers than the big ones which are manipulated by the protected algorithm.
Doing so, we understood how these optimizations work and the power of their underlying ideas.
In this section apply our newly acquired knowledge on the \emph{essence of countermeasures} in order to build the \emph{quintessence of countermeasures}.
Namely, we leverage our findings
to fix Shamir's countermeasure, and to drastically simplify the one of Vigilant, while at the same time transforming it to be infective instead of test-based.

\subsection{Correcting Shamir's Countermeasure}
\label{building:correction}

We saw that Shamir's countermeasure is broken in multiple ways, which has been known for a long time now.
To fix it without denaturing it, we need to verify the integrity of the recombination as well as the ones of the overrings moduli.
We can directly take these verifications from Aumüller \etal's countermeasure.
The result can be observed in Alg.~\ref{alg:fixed-shamir}.

\begin{algorithm*}
\algorithmConfig
\caption{CRT-RSA with a fixed version of Shamir's countermeasure \protect\\ {\smaller (new algorithm contributed in this paper)}}
\label{alg:fixed-shamir}
\Input{Message $M$, key $(p, q, d, i_q)$}
\Output{Signature $M^d \mod N$, or \textsf{error}}
\BlankLine
Choose a small random integer $r$. \;
\BlankLine
$p' = p \cdot r$ \;
$q' = q \cdot r$ \;
\lIf{$p' \not\equiv 0 \mod p$ \Or $q' \not\equiv 0 \mod q$\label{alg:fixed-shamir:p'q'}}{
  \Return\textsf{error}
}
\BlankLine
$S'_p = M^{d \mod \varphi(p')} \mod p'$ \tcp*[r]{Intermediate signature in $\Z_{pr}$}
$S'_q = M^{d \mod \varphi(q')} \mod q'$ \tcp*[r]{Intermediate signature in $\Z_{qr}$}
\lIf{$S'_p \not\equiv S'_q \mod r$\label{alg:fixed-shamir:S'pS'q}}{
  \Return\textsf{error}
}
\BlankLine
$S_p = S'_p \mod p$ \tcp*[r]{Retrieve intermediate signature in $\Z_p$}
$S_q = S'_q \mod q$ \tcp*[r]{Retrieve intermediate signature in $\Z_q$}
\BlankLine
$S = S_q + q \cdot (i_q \cdot (S_p - S_q) \mod p)$ \tcp*[r]{Recombination in $\Z_N$}
\lIf{$S \not\equiv S'_p \mod p$ \Or $S \not\equiv S'_q \mod q$}{
  \Return\textsf{error}
}
\BlankLine
\Return $S$
\end{algorithm*}

The additional tests on line~\ref{alg:fixed-shamir:p'q'} protect against transient faults on $p$ (resp. $q$) while computing $p'$ (resp. $q'$), which would amount to a randomization of $S'_p$ (resp. $S'_q$) while computing the intermediate signatures.
The additional test on line~\ref{alg:fixed-shamir:S'pS'q} verifies the integrity of the intermediate signature computations.

\subsection{Simplifying Vigilant's Countermeasure}
\label{building:simplification}

The mathematical tricks used in the Vigilant countermeasure are very powerful.
Their understanding enabled the optimization of his countermeasure to only need $3$ verifications, while the original version has $9$.
Our simplified version of the countermeasure can be seen in Alg.~\ref{alg:simplified-vigilant}.
Our idea is that it is not necessary to perform the checksum value replacements at lines~\ref{alg:vigilant:R1} and~\ref{alg:vigilant:R2} of Alg.~\ref{alg:vigilant} (see Sec.~\ref{classify:smallrings}).
What is more, if these replacements are not done, then the algorithm's computations carry the CRT-embedded checksum values until the end, and \emph{the integrity of the whole computation can be tested with a single verification} in $\Z_{r^2}$ (line~\ref{alg:simplified-vigilant:verif} of Alg.~\ref{alg:simplified-vigilant}).

This idea not only reduces the number of required verifications, which is in itself a security improvement as shown in Sec.~\ref{classify:test-infec},
but it also optimizes the countermeasure for speed and reduces its need for randomness (the computations of lines~\ref{alg:vigilant:R1} and~\ref{alg:vigilant:R2} of Alg.~\ref{alg:vigilant} are removed).

The two other tests that are left are the ones of lines~\ref{alg:vigilant:verifM'p} and~\ref{alg:vigilant:verifM'q} in Alg.~\ref{alg:vigilant}, which ensure that the original message $M$ has indeed been CRT-embedded in $M'_p$ and $M'_q$.
We take advantage of these two tests to verify the integrity of $N$ both modulo $p$ and modulo $q$ (lines~\ref{alg:simplified-vigilant:M'p} and~\ref{alg:simplified-vigilant:M'q} of Alg.~\ref{alg:simplified-vigilant}).

\begin{remark}
Note that we also made this version of the countermeasure infective, using the transformation method that we exposed in Sec.~\ref{classify:test-infec}.
As we said, any countermeasure can be transformed this way, for instance Alg.~\ref{alg:infective-aumuller} in the Appendix~\ref{sec:infective-aumuller} presents an infective variant of Aumüller \etal's countermeasure.
\end{remark}

\begin{algorithm*}
\algorithmConfig
\caption{CRT-RSA with our simplified Vigilant's countermeasure, under its infective avatar\protect\\ {\smaller (new algorithm contributed in this paper)}}
\label{alg:simplified-vigilant}
\Input{Message $M$, key $(p, q, d_p, d_q, i_q)$}
\Output{Signature $M^d \mod N$, or a random value in $\Z_N$}
\BlankLine
Choose a small random integer $r$. \;
$N = p \cdot q$ \;
\BlankLine
$p' = p \cdot r^2$ \;
$i_{pr} = p^{-1} \mod r^2$ \;
$M_p = M \mod p'$ \;
$B_p = p \cdot i_{pr}$ \;
$A_p = 1 - B_p \mod p'$ \;
$M'_p = A_p \cdot M_p + B_p \cdot (1 + r) \mod p'$ \tcp*[r]{CRT insertion of verification value in $M'_p$}
\BlankLine
$q' = q \cdot r^2$ \;
$i_{qr} = q^{-1} \mod r^2$ \;
$M_q = M \mod q'$ \;
$B_q = q \cdot i_{qr}$ \;
$A_q = 1 - B_q \mod q'$ \;
$M'_q = A_q \cdot M_q + B_q \cdot (1 + r) \mod q'$ \tcp*[r]{CRT insertion of verification value in $M'_q$}
\BlankLine
$S'_p = {M'_p}^{d_p \mod \varphi(p')} \mod p'$ \tcp*[r]{Intermediate signature in $\Z_{pr^2}$}
$S_{pr} = 1 + d_p \cdot r$ \tcp*[r]{Checksum in $\Z_{r^2}$ for $S'_p$}
\BlankLine
$c_p = M'_p + N - M + 1 \mod p$ \label{alg:simplified-vigilant:M'p} \;
\BlankLine
$S'_q = {M'_q}^{d_q \mod \varphi(q')} \mod q'$ \tcp*[r]{Intermediate signature in $\Z_{qr^2}$}
$S_{qr} = 1 + d_q \cdot r$ \tcp*[r]{Checksum in $\Z_{r^2}$ for $S'_q$}
\BlankLine
$c_q = M'_q + N - M + 1 \mod q$ \label{alg:simplified-vigilant:M'q} \;
\BlankLine
$S' = S'_q + q \cdot (i_q \cdot (S'_p - S'_q) \mod p')$ \tcp*[r]{Recombination in $\Z_{Nr^2}$}
$S_r = S_{qr} + q \cdot (i_q \cdot (S_{pr} - S_{qr}) \mod p')$ \tcp*[r]{Recombination checksum in $\Z_{r^2}$}
\BlankLine
$c_S = S' - S_r + 1 \mod r^2$ \label{alg:simplified-vigilant:verif} \;
\BlankLine
\Return $S = S'^{c_pc_qc_S} \mod N$ \tcp*[r]{Retrieve result in $\Z_N$}
\end{algorithm*}

\section{Conclusions and Perspectives}
\label{conclusions}

We studied the existing CRT-RSA algorithm countermeasures against fault-injection attacks, in particular the ones of Shamir's family.
In so doing, we got a deeper understanding of their ins and outs.
We obtained a few intermediate results:
the absence of conceptual distinction between test-based and infective countermeasures,
the fact that faults on the code (skipping instructions) can be captured by considering only faults on the data,
and the fact that the many countermeasures that we studied (and their variations) were actually applying a common protection strategy but optimized it in different ways.
These intermediate results allowed us to describe the design of a high-order countermeasure against our very generic fault model (comprised of randomizing, zeroing, and skipping faults).
Our design allows to build a countermeasure resisting $n$ faults for any $n$ at a very reduced cost (it consists in adding $n-1$ comparisons on small numbers).
We were also able to fix Shamir's countermeasure, and to drastically improve the one of Vigilant, going from $9$ verifications in the original countermeasure to only $3$, removing computations made useless, and reducing its need for randomness, while at the same time making it infective instead of test-based.

Except for those which rely on the fact that the protected algorithm takes the form of a CRT computation, the ideas presented in the various countermeasures can be applied to any modular arithmetic computation.
For instance, it could be done using the idea of Vigilant consisting in using the CRT to embed a known subring value in the manipulated numbers to serve as a checksum.
That would be the most obvious perspective for future work, as it would allow a generic approach against fault attacks and even automatic insertion of the countermeasure.

A study of Giraud's family of countermeasures in more detail would be beneficial to the community as well.

\section*{Acknowledgment}

We would like to thank Antoine Amarilli for his proofreading which greatly improved the editorial quality of our manuscript.

\bibliographystyle{alpha} % ePrint
\begin{footnotesize}
\bibliography{sca}
\end{footnotesize}

\clearpage
% \appendices % IEEE
\appendix % ePrint

\section{Recovering $d$ and $e$ from $(p, q, d_p, d_q, i_q)$}
\label{sec-e_d}

We prove here the following proposition:
\begin{proposition}
It is possible to recover the private exponent $d$ and the public exponent $e$ from the $5$-tuple $(p, q, d_p, d_q, i_q)$
described in Sec.~\ref{sec:CRT_RSA}.
\label{prop:gp_recover_d_from_crt_key}
\end{proposition}
% See: https://svn.comelec.enst.fr/trusted_computing/articles/2014/14_FDTC_HODFA/gp_recover_d_from_crt_key.tex
\begin{proof}
Clearly, $p-1$ and $q-1$ are neither prime, nor coprimes (they have at least $2$ as a common factor).
Thus, proving Prop.~\ref{prop:gp_recover_d_from_crt_key} is not a trivial application of the Chinese Remainder Theorem.
The proof we provide is elementary, but to our best knowledge, it has never been published before.

The numbers $p_1 = \frac{p-1}{\gcd(p-1,q-1)}$ and $q_1 = \frac{q-1}{\gcd(p-1,q-1)}$ are coprime,
but there product is not equal to $\lambda(N)$.
There is a factor $\gcd(p-1,q-1)$ missing, since $\lambda(N) = p_1 \cdot q_1 \cdot \gcd(p-1,q-1)$.

Now, $\gcd(p-1,q-1)$ is expected to be small.
Thus, the following Alg.~\ref{alg:e_d} can be applied efficiently.
In this algorithm, the invariant is that $p_2$ and $q_2$, initially equal to $p_1$ and $p_2$, remain coprime.
Moreover, they keep on increasing whereas $r_2$, initialized to $r_1 = \gcd(p-1,q-1)$, keeps on decreasing till $1$.
%This algorithm is illustrated in Fig.~\ref{fig-gp_recover_d_from_crt_key}.

\begin{algorithm}
\algorithmConfig
\caption{Factorization of $\lambda(N)$ into two coprimes, multiples of $p_1$ and $q_1$ respectively.}
\label{alg:e_d}
%\begin{algorithmic}[1]
\Input{$p_1=\frac{p-1}{\gcd(p-1,q-1)}$, $q_1=\frac{q-1}{\gcd(p-1,q-1)}$ and $r_1=\gcd(p-1,q-1)$}
\Output{$(p_2, q_2)$, coprime, such as $p_2 \cdot q_2 = \lambda(N)$}
\BlankLine
$(p_2, q_2, r_2) \gets (p_1, q_1, r_1)$ \;
\BlankLine
$g \gets \gcd(p_2, r_2)$ \;
\While{$g \neq 1$}{
$p_2 \gets p_2 \cdot g$ \;
$r_2 \gets r_2 / g$ \;
$g \gets \gcd(p_2, r_2)$ \;
}
\BlankLine
$g \gets \gcd(q_2, r_2)$ \;
\While{$g \neq 1$}{
$q_2 \gets q_2 \cdot g$ \;
$r_2 \gets r_2 / g$ \;
$g \gets \gcd(q_2, r_2)$ }
\tcp*[r]{$p_2$, $q_2$ and $r_2$ are now coprime}
%\tcp*[l]{$p_2$, $q_2$ and $r_2$ are now coprime} % Justifies left
$q_2 \gets q_2 \cdot r_2$
\tcp*[r]{$p_2 \gets p_2 \cdot r_2$ would work equally}
$(r_2 \gets r_2 / r_2 = 1)$ \tcp*[r]{For more pedagogy}
\Return{$(p_2, q_2)$}
%\end{algorithmic}
\end{algorithm}

%\begin{figure}
%\begin{center}
%\includegraphics[width=1.0\linewidth]{fig/gp_recover_d_from_crt_key.pdf}
%\end{center}
%\caption{Illustration of Algorithm~\ref{alg-e_d}.}
%\label{fig-gp_recover_d_from_crt_key}
%\end{figure}

Let us denote $p_2$ and $q_2$ the two outputs of Alg.~\ref{alg:e_d}, we have:
\begin{compactitem}
\item $d_{p_2} = d_p \mod p_2$, since $p_2 | (p-1)$;
\item $d_{q_2} = d_q \mod q_2$, since $q_2 | (q-1)$;
\item $i_{12} = {p_2}^{-1} \mod q_2$, since $p_2$ and $q_2$ are coprime.
\end{compactitem}
We can apply Garner's formula to recover $d$:
\begin{align}
d = d_{p_2} + p_2 \cdot (( i_{12} \cdot ( d_{q_2} - d_{p_2} )) \mod q_2) \enspace.
\end{align}
By Garner, we know that $0\leq d < p_2 \cdot q_2 = \lambda(N)$,
which is consistent with the remark made in the last sentence of Sec.~\ref{sec:RSA}.
% See: https://svn.comelec.enst.fr/trusted_computing/articles/2014/14_FDTC_HODFA/gp_recover_d_from_crt_key/tests.gp

Once we know the private exponent $d$, the public exponent $e$ can be computed as the inverse of $d$ modulo $\lambda(N)$.
\end{proof}

\section{Infective Aumüller CRT-RSA} % ePrint
\label{sec:infective-aumuller}

The infective variant of Aumüller protection against CRT-RSA is detailed in Alg.~\ref{alg:infective-aumuller}.

\begin{algorithm}
\algorithmConfig
\caption{CRT-RSA with Aumüller \etal's countermeasure\footref{fn:no-spa}, under its infective avatar\protect\\ {\smaller (new algorithm contributed in this paper)}}
\label{alg:infective-aumuller}
\Input{Message $M$, key $(p, q, d_p, d_q, i_q)$}
\Output{Signature $M^d \mod N$, or a random value}
\BlankLine
Choose a small random integer $r$. \;
\BlankLine
$p' = p \cdot r$ \;
$c_1 = p' + 1 \mod p$ \;
\BlankLine
$q' = q \cdot r$ \;
$c_2 = q' + 1 \mod q$ \;
\BlankLine
$S'_p = M^{d_p \mod \varphi(p')} \mod p'$ \tcp*[r]{Intermediate signature in $\Z_{pr}$}
$S'_q = M^{d_q \mod \varphi(q')} \mod q'$ \tcp*[r]{Intermediate signature in $\Z_{qr}$}
\BlankLine
$S_p = S'_p \mod p$ \tcp*[r]{Retrieve intermediate signature in $\Z_p$}
$S_q = S'_q \mod q$ \tcp*[r]{Retrieve intermediate signature in $\Z_q$}
\BlankLine
$S = S_q + q \cdot (i_q \cdot (S_p - S_q) \mod p)$ \tcp*[r]{Recombination in $\Z_N$}
$c_3 = S - S'_p + 1 \mod p$ \;
$c_4 = S - S'_q + 1 \mod q$ \;
\BlankLine
$S_{pr} = S'_p \mod r$ \tcp*[r]{Checksum of $S_p$ in $\Z_{r}$}
$S_{qr} = S'_q \mod r$ \tcp*[r]{Checksum of $S_q$ in $\Z_{r}$}
$c_5 = {S_{pr}}^{d_q \mod \varphi(r)} - {S_{qr}}^{d_p \mod \varphi(r)} + 1 \mod r$ \;
\BlankLine
\Return $S^{c_1c_2c_3c_4c_5}$
\end{algorithm}

\end{document}